\documentclass[letterpaper, 10 pt, conference]{ieeeconf}  

\IEEEoverridecommandlockouts                              
\usepackage{amsmath,amssymb,amsfonts}
\newtheorem{thm}{Theorem}

\newtheorem{remark}{\bf Remark}
\usepackage{flushend}
\usepackage{graphicx}
\usepackage{subfigure}
\onecolumn

\title{\LARGE \bf
Humans-in-the-Building: Getting Rid of Thermostats for Optimal Thermal Comfort Control in Energy Management Systems
}

\author{Jiali Wang, Yang Tang, \IEEEmembership{Fellow, IEEE}, and Luca Schenato, \IEEEmembership{Fellow, IEEE}
\thanks{This work was supported by the National Natural Science
Foundation of China (No. 62293502).}
\thanks{J. Wang is with Key Laboratory of Smart Manufacturing in Energy Chemical Process, Ministry of Education, East China University of Science and Technology, Shanghai 200237, China, and also with the Department of Information Engineering, University of Padova, 35131 Padova, Italy (e-mail: jialiwang@mail.ecust.edu.cn).}%
\thanks{Y. Tang is with Key Laboratory of Smart Manufacturing in Energy Chemical Process, Ministry of Education, East China University of Science and Technology, Shanghai 200237, China (e-mail: yangtang@ecust.edu.cn).}%
\thanks{L. Schenato is with the Department of Information Engineering, University of Padova, 35131 Padova, Italy (email: schenato@dei.unipd.it).}%
}

\begin{document}

\linespread{1.25}

\maketitle

\thispagestyle{empty}
\pagestyle{empty}

\begin{abstract}
Given the widespread attention to individual thermal comfort, coupled with significant energy-saving potential inherent in energy management systems for optimizing indoor environments, this paper aims to introduce advanced ``Humans-in-the-building" control techniques to redefine the paradigm of indoor temperature design. Firstly, we innovatively redefine the role of individuals in the control loop, establishing a model for users' thermal comfort and constructing discomfort signals based on individual preferences. Unlike traditional temperature-centric approaches, ``thermal comfort control" prioritizes personalized comfort. Then, considering the diversity among users, we propose a novel method to determine the optimal indoor temperature range, thus minimizing discomfort for various users and reducing building energy consumption. Finally, the efficacy of the ``thermal comfort control" approach is substantiated through simulations conducted using Matlab.
\end{abstract}
\begin{keywords}
Humans-in-the-building, thermal comfort, discomfort signals, building energy consumption	
\end{keywords}

\section{INTRODUCTION}
The thermal comfort perception of users is of paramount importance in indoor environmental quality, attracting numerous researchers for in-depth exploration~\cite{kizilkale2019integral, eichler2018humans, tavakkoli2020bonus}. According to the internationally recognized definition, thermal comfort refers to a psychological state of satisfaction with the surrounding temperature. However, some researchers posit that this definition may be an outcome of perceptual processes, viewing thermal comfort as an environmental attribute correlated with physical climate and heating, ventilation, and air conditioning (HVAC) systems control~\cite{heijs1994dependent}. On the other hand, some researchers consider it a subjective sensation, asserting that there is no perfect combination of conditions that can make everyone feel comfortable, even under optimal indoor temperature conditions, where only fewer than $70\%$ of individuals may experience comfort~\cite{markus1980buildings}. Consequently, thermal comfort lacks an absolute definition and depends on indoor and outdoor temperatures, user expectations, and each user's tolerance threshold for temperature.

Acquiring insights into users' diverse perceptions of thermal comfort amidst fluctuations in room temperature is crucial for comprehending their physiological responses~\cite{fanti2015district}. Consequently, this process aids in assessing the impact of individual differences on determining the optimal room temperature. The overall comfort of users is directly determined by their perceptions of negative and positive thermal sensations~\cite{al2022quantifying}. More precisely, the concept of comfort can be viewed as a form of ``complaint": users articulate a ``complaint" about negative thermal sensation when they sense the room temperature is below their ideal comfort level. Conversely, a ``complaint" about positive thermal sensation is expressed when users feel the room temperature is above their ideal level. We categorize the varied responses of individuals to temperature as personal comfort signals. By employing these signals, it is possible to derive an individually customized optimal indoor temperature, considering individual variations, thereby augmenting overall comfort and productivity~\cite{jung2019human}.

Furthermore, integrating residents' comfort perceptions into the energy management process is vital to enhance energy management systems and promote residents' well-being~\cite{gulbinas2015segmentation}. This specifically entails ensuring user comfort while simultaneously reducing energy consumption. However, finding a balance between lowering the energy performance of buildings and enhancing comfort for different users presents a challenge. The main objectives of humans-in-the-building HVAC system operation studies involve minimizing energy consumption to the greatest extent without compromising thermal comfort and optimizing thermal comfort for a diverse range of users~\cite{squercinatemperature}. Given the significant rise in energy costs within the industry and service sectors, the consideration of both reducing energy consumption and individual comfort emerges as a goal with mutual benefits~\cite{fanger1970thermal}.

Currently, research on humans-in-the-building primarily focuses on how to reduce the energy consumption of buildings, without fully taking into account the thermal comfort of each user. In contrast, our approach centers on individual comfort, recognizing user diversity, and aims to minimize energy consumption significantly. 
The contribution of this paper is twofold. The first contribution is that, compared with \cite{inoue2019weak} and \cite{sadowska2023predictive}, we establish a thermal comfort model based on users' sensations from a new perspective, and generate discomfort signals according to individual differences to achieve the optimal indoor temperature.
The second contribution is that, based on user preferences, we propose a novel method to determine the optimal indoor temperature range, aiming to minimize discomfort for various users and reduce building energy consumption. This is an innovative approach as it directly considers the individual differences of each user.

This paper is organized as follows. The user comfort modeling and signaling are given in Section II. In Section III, we introduce the comfort control design. The simulations and results are presented in Section IV. Finally, Section V concludes this paper and outlines future work.

\section{USER COMFORT MODELING AND SIGNALING}

\begin{figure}
	\begin{center}
		\includegraphics[width=0.5\textwidth]{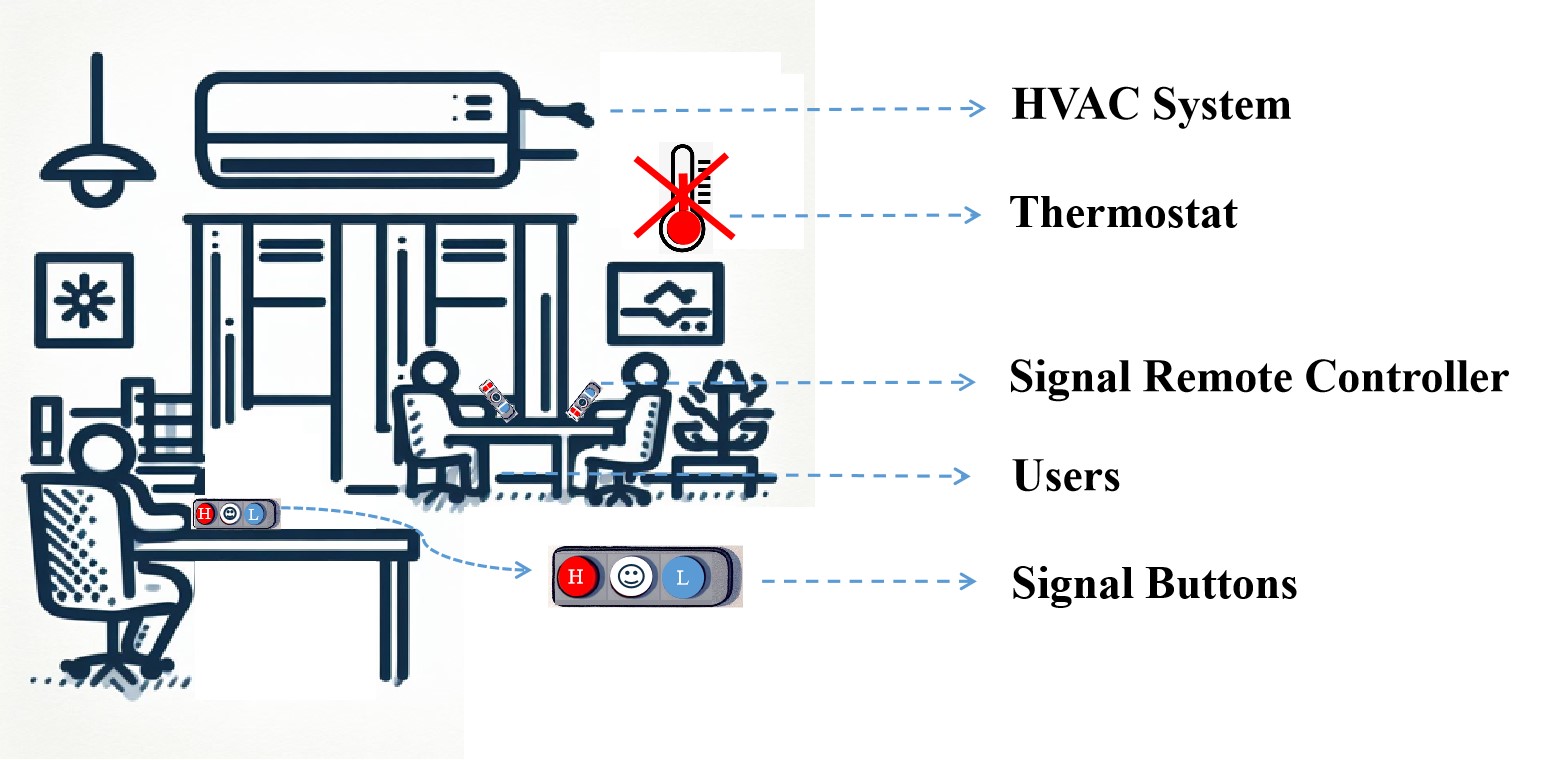}
		\caption{Room Prospective}
		\label{room}
	\end{center}
\end{figure}

Assuming there are $N$ occupants and one HVAC system in a room, as shown in Fig. \ref{room}. Each user holds a temperature sensation signal remote controller, which is equipped with three buttons representing cold, normal, and warm signals, respectively. We use $T_{ext}$ and $T_r$ to represent outdoor temperature and indoor temperature, respectively. Each user $i~(i =1,..., N)$ has their own ideal comfort temperature, denoted as $T_i^*$. Similar to previous studies \cite{al2022quantifying, aryal2018energy}, we utilize Gaussian mixture models to define the discomfort function $\tilde{f}$ for user $i$. The formulation is given by
\begin{align}
	\tilde{f}_i\left( T_r \right) =\left\{ \begin{array}{c}
		1-e^{-\frac{\left( T_r-T_{i}^{*} \right) ^2}{\sigma_i ^2}}, \quad T_r<T_{i}^{*},\\
		e^{-\frac{\left( T_r-T_{i}^{*} \right) ^2}{\sigma_i ^2}}-1, \quad T_r>T_{i}^{*},
	\end{array} \right.
\end{align}
where $\sigma_i$ represents the temperature sensitivity of user $i$. 
As the temperature increases, the user's discomfort decreases gradually from $+1$ to $-1$, indicating that $ \tilde{f}_i\left( T_r \right) $ is monotonically decreasing. Therefore, there exists an optimal room temperature $T_r^*$ where $\tilde{f}_i\left( T_r^* \right) = 0 $.

In practical scenarios, we are more concerned about the degree of discomfort felt by users rather than whether they feel cold or hot. Therefore, we use $ f_i $ to represent the absolute value of the user's discomfort function, i.e., 
\begin{align}
	f_i\left( T_r \right) =|\tilde{f}_i\left( T_r \right) |.
\end{align}
Clearly, as the temperature rises, the absolute discomfort function $f_i$ exhibits a gradual decrease from $1$ to $0$, followed by a subsequent increase to $1$, as depicted in Fig. \ref{four_users_abs}.
\begin{figure}
	\centering
		\includegraphics[width=0.5\textwidth]{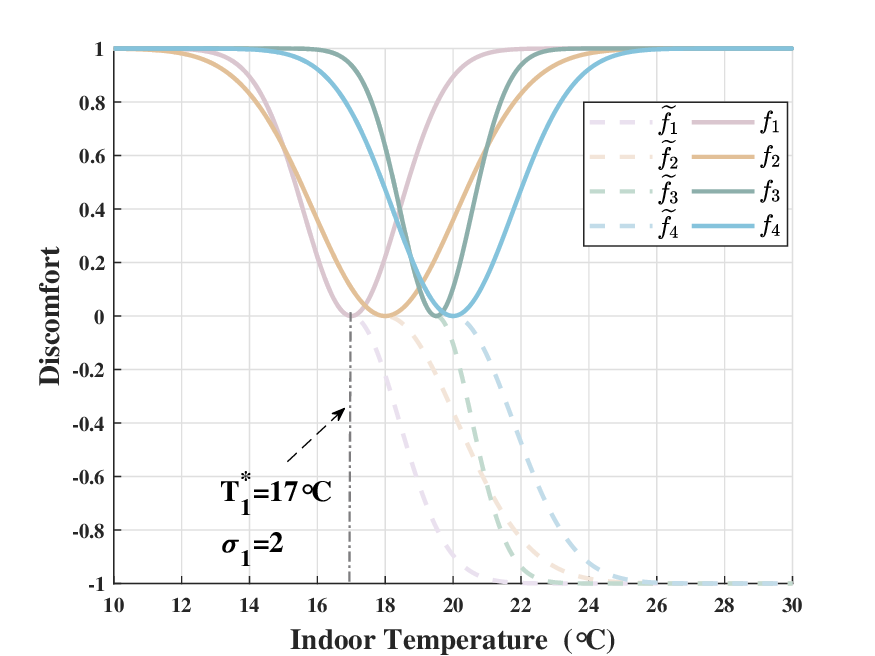}
	\caption{The discomfort $\tilde{f}_i$ and absolute discomfort $f_i$ of four users, $T^*=[17, 18, 19.5, 20]$, $\sigma_i=[2, 3, 1.5, 2.5]$.}
	\label{four_users_abs}
\end{figure}

\begin{figure}
	\begin{center}
		\includegraphics[width=0.5\textwidth]{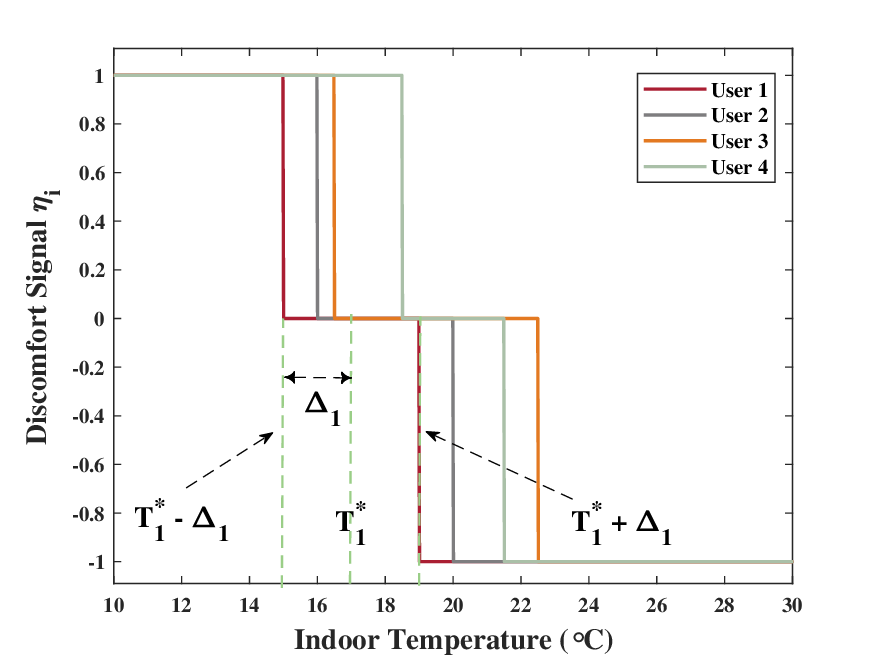}
		\caption{Discomfort feedback signal of four users, $T^*=[17, 18, 19.5, 20]$,  ${\Delta_i} = [2, 2, 3, 1.5]$.}
		\label{discomfort_signal}
	\end{center}
\end{figure}

Each user $i$ has a certain thermal comfort tolerance $\Delta_i$, defining the acceptable indoor temperature range $[T_i^*-\Delta_i, T_i^*+\Delta_i]$ for user $i~(i =1,..., N)$. To comprehensively consider individual thermal comfort, we assume that each user carries a remote controller, allowing real-time expression of their comfort and discomfort signals, as shown in Fig. \ref{discomfort_signal}. Specifically, for user $i$, pressing the red, blue, and white buttons respectively indicates feeling cold ($+1$), feeling hot ($-1$), and feeling comfortable ($0$) within their acceptable temperature range. The ideal temperature for each user remains constant, making the discomfort signal dependent on both indoor temperature $T_r$ and thermal comfort tolerance $\Delta_i$ of all users. Subsequently, we represent the discomfort signal of user $i$ as $\eta_i$, i.e.,
\begin{align}
	\eta _i\left( T_r,\Delta _i \right) 
	=\left\{ \begin{array}{l}
		+1,\ ~T_{r}<T_{i}^*-\Delta _i,\\
	~	0, ~~ ~T_{i}^*-\Delta _i \le T_{r} \le T_{i}^*+\Delta _i,\\
		-1,\ ~T_{r}>T_{i}^*+\Delta _i.
	\end{array} \right.
\end{align}

To achieve the optimal indoor temperature, it is essential to consider the discomfort signals from all users comprehensively. Therefore, we define the sum of discomfort signals from all users as $h$:
\begin{equation}
\label{sum_h}
h\left( T_r,\{\Delta _i\} \right) =\sum_i^N{\eta _i\left( T_r,\Delta _i \right)} \in \mathbb{Z},
\end{equation}
where $\{\Delta_i\}$ represents the set of thermal comfort tolerances for all users. Regarding \( h \), we deduce the following theorem.


\begin{figure}
	\centering
	\subfigure[Symmetrical scenario: $\Delta_1= \Delta_2=2, \Delta_3=3, \Delta_4=1.5$.]{
		\includegraphics[width=0.45\textwidth]{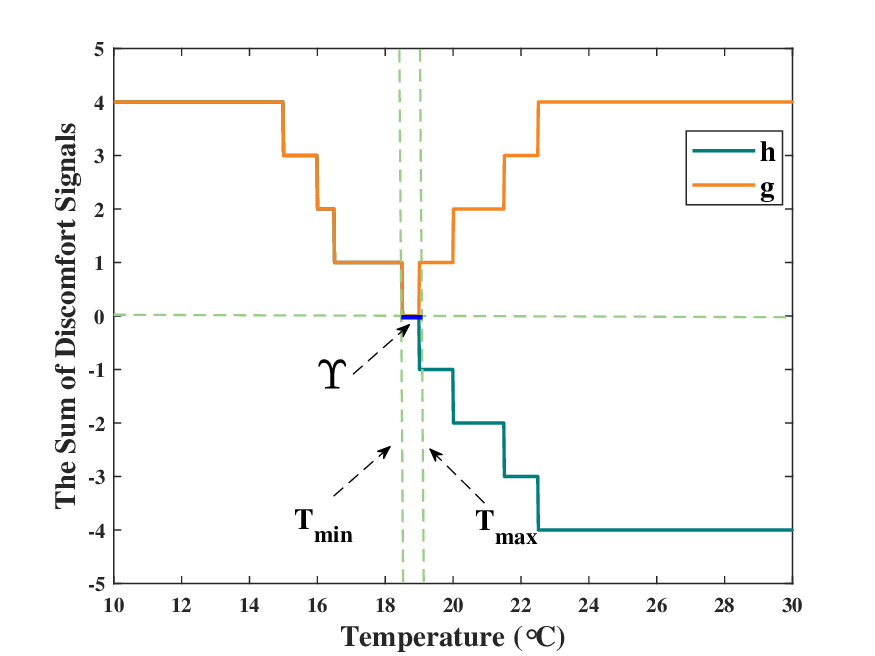}
		\label{g}
	}
	\subfigure[Asymmetrical scenario: $\Delta_1=0.5, \Delta_2=1, \Delta_3=2, \Delta_4=1.5$.]{
		\includegraphics[width=0.45\textwidth]{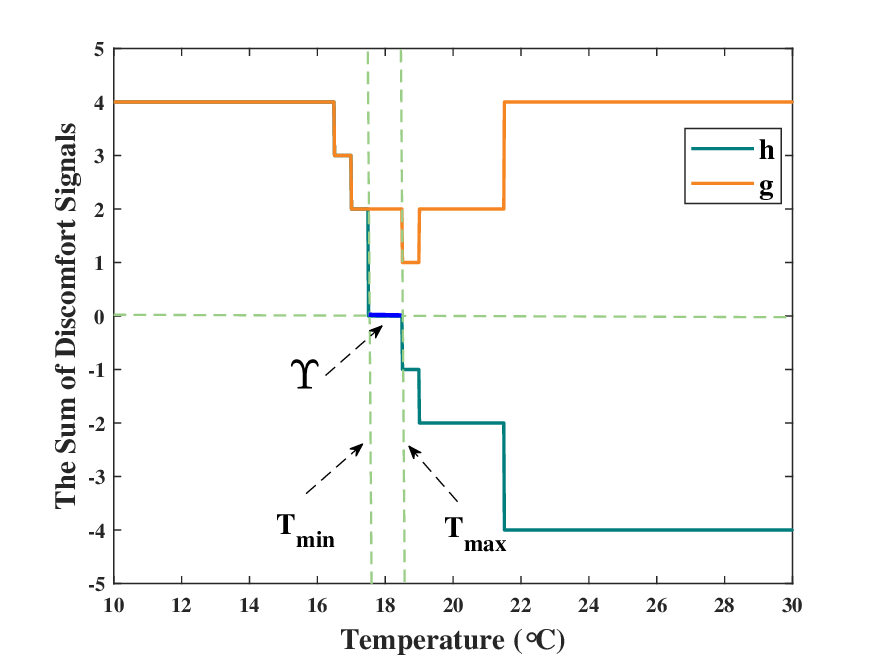}
		\label{asy}
	}
	\caption{The sum and absolute sum of user discomfort signals.}
\end{figure}

\begin{thm}
In the humans-in-the-building control, the total discomfort signal sum $h$ is a monotonically decreasing stepwise function, and
\begin{align*}
\underset{T_r\rightarrow -\infty}{\lim}h\left( T_r,\{\Delta _i\} \right) &= N,\\
\underset{T_r\rightarrow +\infty}{\lim}h\left( T_r,\{\Delta _i\} \right) &= -N.
\end{align*}
\end{thm}
\begin{proof}
For user $i$, when the room temperature $T_r$ is lower than the user's ideal temperature $T_i^*$ minus a threshold $\Delta_i$, the discomfort signal is $+1$; when the room temperature $T_r$ is higher than the user's ideal temperature $T_i^*$ plus a threshold $\Delta_i$, the discomfort signal is $-1$; when the room temperature $T_r$ is within the threshold range of the user's ideal temperature $T_i^*$, the discomfort signal for user $i$ is $0$. Therefore, as the room temperature $T_r$ increases, the function $\eta_i(T_r, \Delta_i)$ decreases monotonically. Obviously, the sum of discomfort signals for all users $h(T_r, \{\Delta_i\})$ decreases monotonically in a stepwise fashion with the increase in room temperature $T_r$.
\end{proof}

\begin{remark}
For instance, consider a room with $4$ users, assuming the users' ideal room temperature set is $T^*=[17, 18, 19.5, 20]$, and the comfort tolerance set is ${\Delta_i} = [2, 2, 3, 1.5]$. With the variation of indoor temperature, the sum of discomfort signals for all users can be obtained, as shown by the green curve in Fig. \ref{g}. It is evident that the sum $h$ of these signals monotonically decreases in a stepwise fashion.
\end{remark}

\begin{figure}
	\begin{center}
		\includegraphics[width=0.5\textwidth]{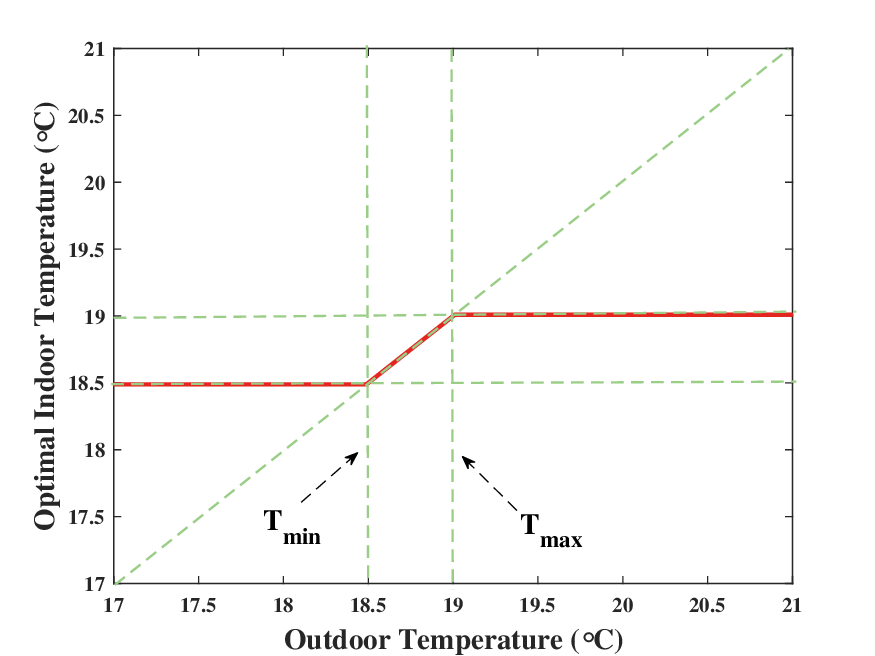}
		\caption{Optimal room temperature selection with variations in outdoor temperature.}
		\label{optimal_temp}
	\end{center}
\end{figure}

To determine the total number of users feeling discomfort at different temperatures, we define the sum of absolute values of discomfort signals as follows
\begin{equation}
g\left( T_r,\{\Delta _i\} \right) = \sum_i^N\left|{\eta _i\left( T_r,\Delta _i \right)} \right|.
\label{g_t_r}
\end{equation}
In Fig. \ref{g}, we observe that the orange curve either coincides with the cyan curve or is symmetrical to the cyan curve. This implies that the absolute value of the sum of discomfort signals is either equal or opposite between the two curves. However, there are also asymmetric situations between them, as shown in Fig. \ref{asy}. When the indoor temperature is between $18^\circ$C degrees and $18.5^\circ$C, two users feel cold $(+1)$, one user feels hot $(-1)$, and one user feels comfortable. In this case, the sum of discomfort signals is $1$, while the total number of users feeling discomfort is $3$. Then, $h$ and $g$ are neither equal nor symmetrical. To ensure fairness among users, we set the optimal indoor temperature range to the temperature obtained when $h=0$.
Therefore, based on the comfort tolerance set of all users, the optimal indoor temperature range needs to satisfy the following condition: 
\begin{align*}
\Upsilon := \left[ T_{\min}\left( \{\Delta _i\} \right) ,T_{\max}\left( \{\Delta _i\} \right) \right] =\{T_r:h\left( T_r,\{\Delta _i\} \right) =0\},
\end{align*}
where $T_{{\max}}$ and $T_{{\min}}$ represent the maximum and minimum optimal indoor temperatures, respectively. Therefore, regarding the optimal indoor temperature range, we can establish the following theorem:

\begin{thm}
If $\sup \limits_{i}\{T_i^*-\Delta_i\} < \inf \limits_{j}\{T_j^*-\Delta_j\} (i,j=\{1,2,...,N\})$, the optimal indoor temperature range $[T_{\min}, T_{\max}]$ satisfies
\begin{align}
    T_{\min} &= \sup \limits_{i}\{T_i^*-\Delta_i\}, \notag \\
    T_{\max} &= \inf \limits_{j}\{T_j^*+\Delta_j\}.
    \label{T_proof}
\end{align}
For any $T_r \in \Upsilon = \left[T_{\min}, T_{\max}\right]$, the sum of absolute values of discomfort signals from all users $g\left( T_r,\{\Delta _i\} \right)$ satisfies
\begin{align*}
g\left( T_r,\{\Delta _i\} \right)=0.
\end{align*}
\end{thm}
\begin{proof}
By definition, given that $\sup \limits_{i}\{T_i^*-\Delta_i\} < \inf \limits_{j}\{T_j^*+\Delta_j\}$, it implies that there exists an intersection in the acceptable temperature ranges for all users, denoted as $[\sup \limits_{i}\{T_i^*-\Delta_i\}, \inf \limits_{j}\{T_j^*+\Delta_j\}]$. The optimal indoor temperature range is the common intersection that satisfies all users. Consequently, this range can be determined, providing the maximum and minimum values of indoor temperature as given in Eq. (\ref{T_proof}).
For any $T_r \in \Upsilon$, where $T_{\min} \leq T_r \leq T_{\max}$, it holds true that for any user $i$, $T_r$ is within the range $T_i^*-\Delta_i < T_r < T_i^*+\Delta_i$.
Now, according to Eq. (\ref{g_t_r}), since $T_r$ lies within the range $[\sup \limits_{i}\{T_i^*-\Delta_i\}, \inf \limits_{j}\{T_j^*+\Delta_j\}]$ for all $i, j$, it guarantees that each term $\left|{\eta _i\left( T_r,\Delta _i \right)} \right|$ will be equal to zero when $T_r$ is within the acceptable temperature range for all users.
\end{proof}

\begin{figure}
	\begin{center}
		\includegraphics[width=0.5\textwidth]{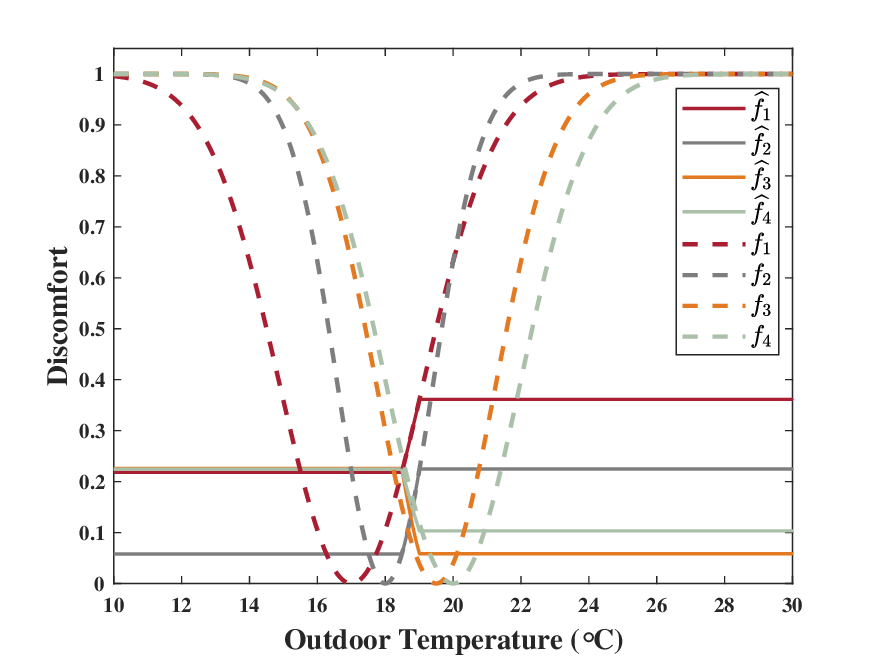}
		\caption{Solid lines: taking optimal indoor temperature, the absolute values of discomfort for the four users; Dash lines: the absolute discomfort for four users.}
		\label{discomfort_f}
	\end{center}
\end{figure}

\section{COMFORT CONTROL DESIGN}
Based on the user comfort modeling and response presented in the previous section, here we provide the control design to be implemented by the energy management system that minimizes energy consumption while maintaining a optimal thermal comfort level.

To minimize energy consumption to the greatest extent, we can formulate it as the following optimization problem:
\begin{align}
&T_{r}^{*}=\underset{T_r}{\min}P\left( T_{ext},T_r \right) =\mu|T_{ext}-T_r|, \notag \\
&~~~~~~~~\text{s.t.} ~~h\left( T_r,\{\Delta _i\} \right) =0,
\end{align}
where $P$ represents the power consumption of the energy management systems, and $\mu$ is a constant.
Then, the optimal indoor temperature $T_{r}^{*}$ is defined as follows
\begin{align}
	T_{r}^{*}\left( T_{ext},\{\Delta _i\} \right) 
	=\left\{ \begin{array}{l}
		T_{\min} ,~T_{ext}<T_{\min},\\
		T_{ext}, ~~ T_{\min}<T_{ext}<T_{\max},\\
		T_{\max} , ~T_{ext}>T_{\max}.
	\end{array} \right.
\end{align}
In Fig. \ref{optimal_temp}, the optimal indoor temperature can be observed adjusting with the change in outdoor temperature.
\begin{remark}
The HVAC system optimizes energy conservation by adjusting the indoor temperature based on outdoor conditions. If the outdoor temperature $T_{\text{ext}}$ is below ${T_{{\min}}}$, the indoor temperature is set to the minimum of the optimal range, ${T_{{\min}}}$. When the outdoor temperature is within the range of ${T_{{\min}}}$ to ${T_{{\max}}}$, the indoor temperature follows the outdoor trend. When the outdoor temperature exceeds ${T_{{\max}}}$, the indoor temperature is adjusted to ${T_{{\max}}}$.
\end{remark}

$T_{r}^{*}$ is influenced by both the outdoor temperature and the set of thermal comfort tolerances of users. Let
$T_{r}^{*}=\Psi \left( T_{ext},\{\Delta _i\} \right). $
Then, we can obtain the absolute value $\widehat{f_i}$ of discomfort expected by user $i$ as
\begin{align}
\widehat{f_i}\left( T_{ext} \right) :=&f_i\left( T_{r}^{*} \right)
\notag \\
=&f_i\left( \Psi \left( T_{ext},\{\Delta _i\} \right) \right). 
\end{align}
According to the optimal indoor temperature, the absolute change in discomfort expectations for each user is illustrated in Fig. \ref{discomfort_f}.

For user $i$, it is only possible to determine their own desired indoor temperature $T_i^*$ and comfort tolerance $\Delta_i$, and it is not possible to determine the comfort tolerance of other users $\Delta_{-i}$. Therefore, in order to maximize their own thermal comfort, the utility value $u_i$ for user $i$ can be defined as:
\begin{align}
u_i\left( \Delta _i,\Delta _{-i} \right) :=&\underset{T_{ext}}{\max}f_i\left( \Psi \left( T_{ext},\{\Delta _i\} \right) \right) \notag \\
=&\max \left\{ f_i\left( T_{\min} \right) ,f_i\left( T_{\max} \right) \right\} .
\end{align}
However, for the building energy management system, it is necessary to consider the discomfort of all users. Therefore, when the comfort tolerance $\Delta_i$ of each user is not a fixed value, the building energy management system considers the worst scenario, which is the maximum discomfort among users at each step of the comfort tolerance $\Delta_i$ variation. Then, the maximum overall thermal discomfort is
\begin{align}
	y=\max u_i\left( \Delta _i,\Delta _{-i} \right) .
\end{align}

\begin{figure}
	\centering
	\subfigure[$\Delta=0$]{
		\includegraphics[width=0.45\textwidth]{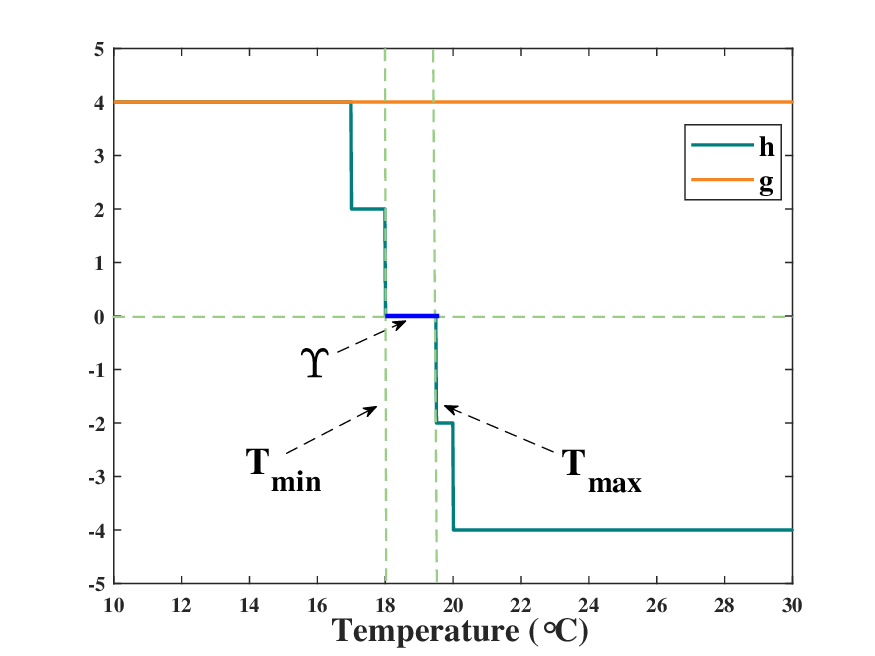}
		\label{different_sum_dis1}
	}
	\subfigure[$\Delta=1$]{
		\includegraphics[width=0.45\textwidth]{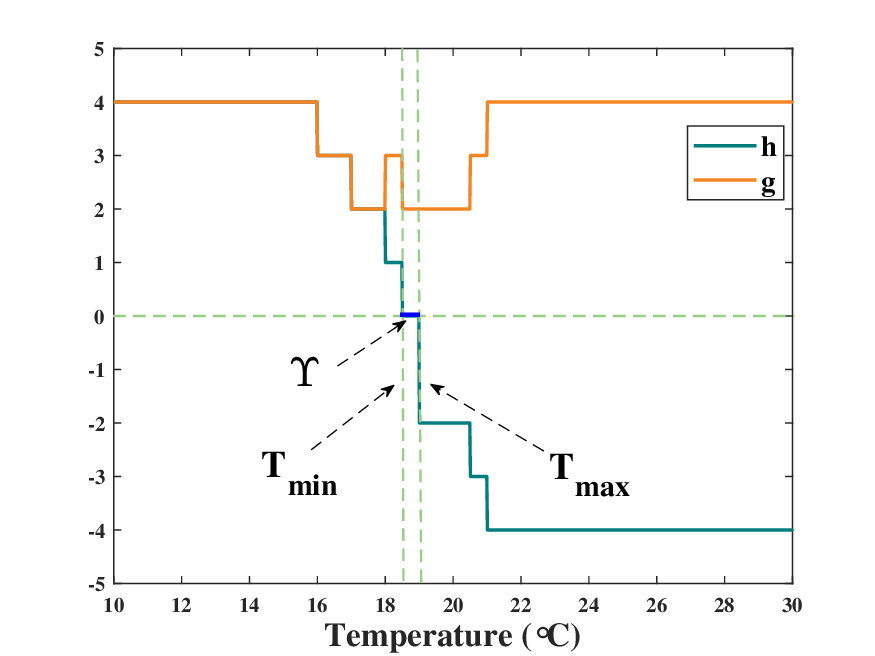}
		\label{different_sum_dis2}
	}
	\subfigure[$\Delta=2$]{
		\includegraphics[width=0.45\textwidth]{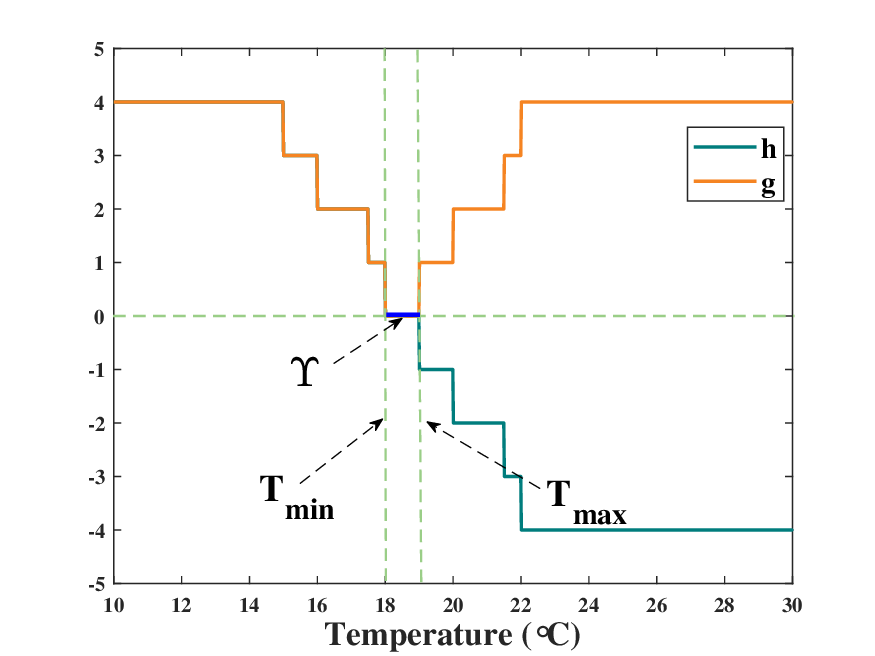}
		\label{different_sum_dis3}
	}
	\subfigure[$\Delta=3$]{
		\includegraphics[width=0.45\textwidth]{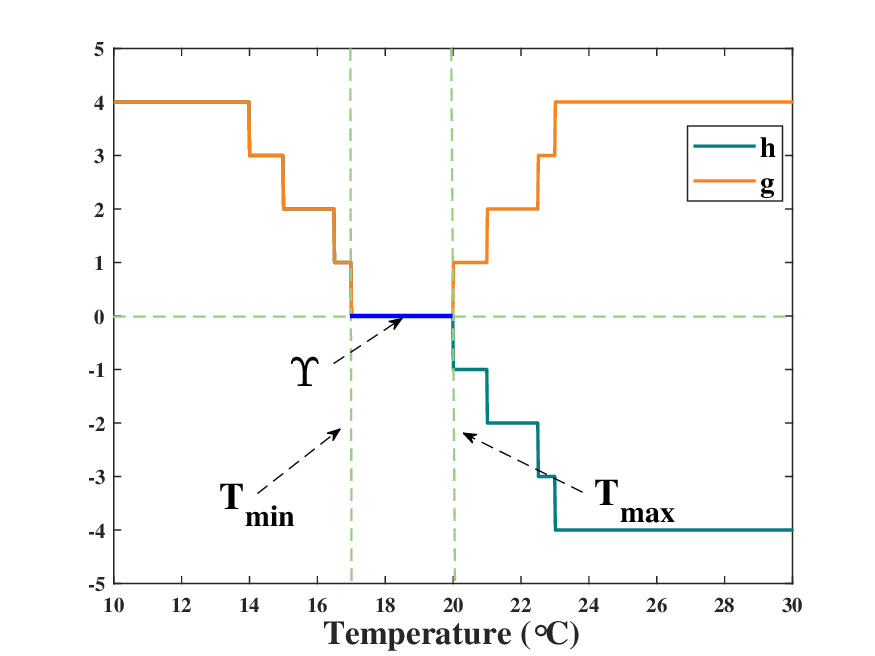}
		\label{different_sum_dis4}
	}
	\caption{The sum of discomfort signals for four users under different thermal comfort tolerances, where $\Delta_i=\Delta~(\forall \ i=1,2,3,4)$.}
	\label{different_sum_dis}
\end{figure}

Additionally, the variation in room temperature is influenced by the heat exchange between the room and the external environment, as well as the presence of additional heat sources (or cooling sources). Next, the control of the indoor temperature by the HVAC system is considered through the following heat transfer differential equation:
\begin{align}
\dot{T_r}(t)=-c\left( T_r(t)-T_{ext}(t) \right) +w(t),
\end{align}
where $ T_r(t) $ and $ T_{ext}(t) $ are the room and outdoor temperature at time $ t $, respectively. $ c $ is a positive constant representing the heat exchange coefficient between the room and the external environment, determining the rate at which the room temperature adjusts to the external temperature, and $ w(t) $ represents the temperature control provided by the HVAC system at time $ t $. Therefore,  Eq. (\ref{sum_h}) can be rewritten at time $t$ as:
$h\left( t \right) =\sum_i{\eta _i\left( T_r\left( t \right) \right)}$,
and the control of the HVAC system is:
\begin{align}
	\label{w_t}
w(t)=-kh\left( t \right), 
\end{align}
where $k$ is the control gain.
\begin{remark}
Derived from individual preferences, discomfort signals are generated based on users' ideal indoor temperature and comfort tolerance. The building energy management system aims to minimize energy consumption by optimizing to achieve zero total discomfort signals for all users, thus determining the optimal indoor temperature range. The HVAC system dynamically adjusts heating or cooling based on outdoor temperature fluctuations and the optimal room temperature from Eq. (\ref{w_t}). The objective is to minimize user discomfort while maximizing energy efficiency.
\end{remark}

\addtolength{\textheight}{0cm}   




\section{SIMULATIONS AND RESULTS}

In this section, we aim to investigate the variation in the optimal indoor temperature and users' discomfort when the users' thermal comfort tolerances change. We consider a room with four users, i.e., $N=4$. Their desired indoor temperatures are $T^* = [17, 18, 19.5, 20]$, and the values of their temperature sensitivities are $\sigma = [3, 2, 2.5, 2.8]$. Setting the range of each user $i$'s variable thermal comfort tolerance as: $\Delta_i \in [0, 3]$.
To implement this, we discretized the tolerance range with a step size of $\delta_i = 0.03$.

Considering the thermal comfort tolerance range, we observe the continuous variation of the total discomfort signal and its absolute values for all users. Specifically, we highlight four scenarios ($\Delta_i =\Delta = 0, 1, 2, 3$) in Fig. \ref{different_sum_dis}. Notably, in cases where $\Delta_i =\Delta = 0$ and $\Delta_i =\Delta = 1$, $h \neq g$ when $h = 0$. This suggests that, even within the optimal indoor temperature range, some users remain dissatisfied. However, achieving $h = 0$ implies fairness, indicating an equal number of users experiencing cold and hot sensations. For $\Delta_i =\Delta = 2$ and $\Delta_i =\Delta = 3$, $h = g$ when $h = 0$. This implies that under optimal indoor temperature conditions, all users feel satisfied. In such instances, for $h > 0$, the sum of discomfort signals and its absolute value satisfy $g = h$, while for $h < 0$, $g = -h$, demonstrating symmetry.

Due to the goal of the building energy management system, the optimal indoor temperature range varies accordingly with changes in users' thermal comfort tolerance, as illustrated in Fig. \ref{min_max_temp}. When users have a thermal comfort tolerance of $3$, their acceptable temperature range expands, leading to a broader optimal indoor temperature range $[17, 20]$ where the sum of discomfort signals is zero. Conversely, with a thermal comfort tolerance of $1.5$, the optimal indoor temperature range narrows, with maximum and minimum values converging around $18.5^\circ$C.

\begin{figure}
	\begin{center}
		\includegraphics[width=0.5\textwidth]{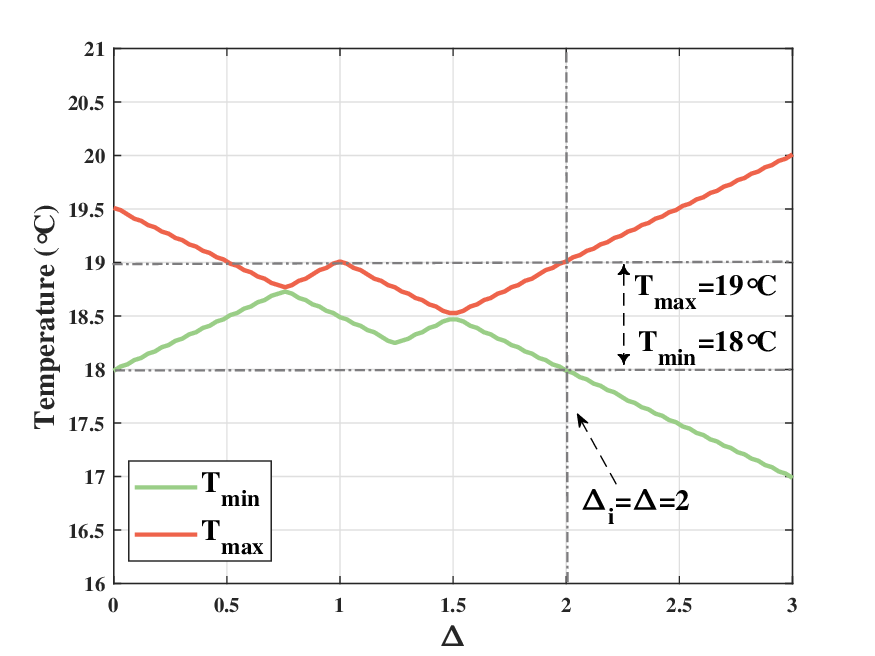}
		\caption{Maximum and Minimum Temperatures with different $\Delta$, where $\Delta_i=\Delta~(\forall \ i=1,2,3,4)$.}
		\label{min_max_temp}
	\end{center}
\end{figure}

\begin{figure}
	\centering
		\includegraphics[width=0.5\textwidth]{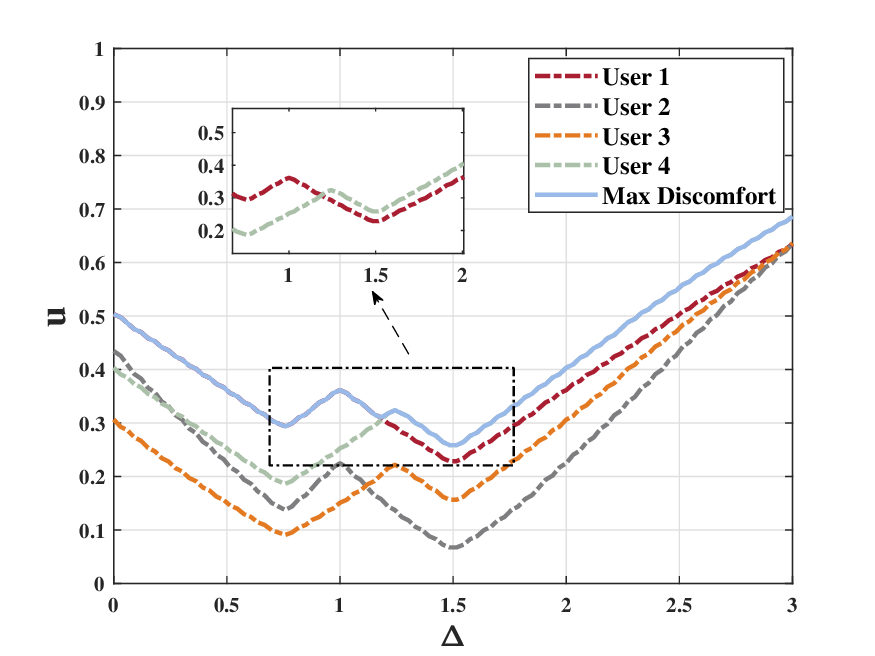}
		\caption{Dash lines: individual discomfort for each user varies with different $\Delta$, where $\Delta_i=\Delta~(\forall \ i=1,2,3,4)$; Solid line: the maximum discomfort of four users.}
		\label{individual_dis}
\end{figure}
\begin{figure}
	\begin{center}
		\includegraphics[width=0.5\textwidth]{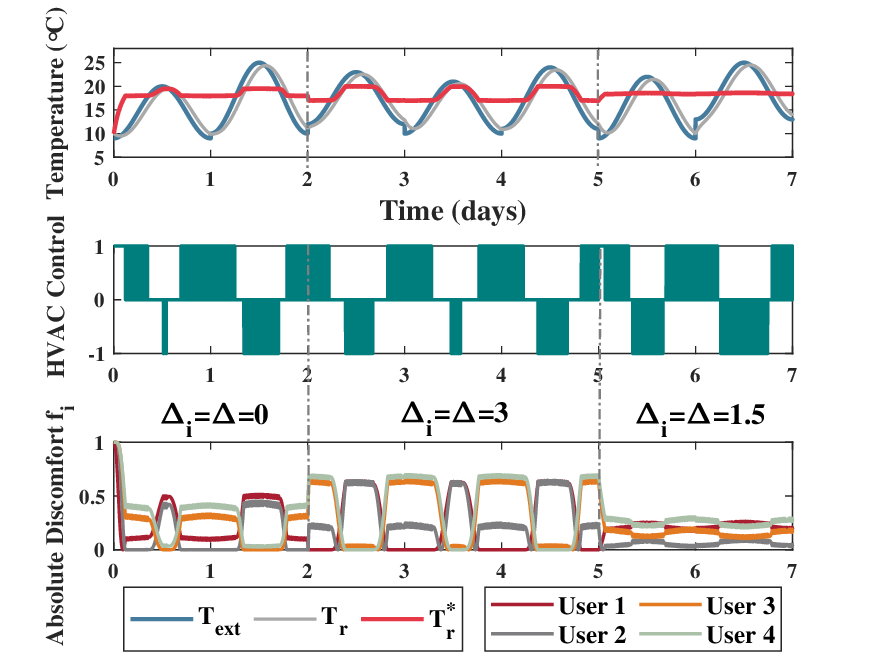}
		\caption{Top and Middle: the variation of indoor temperature with outdoor temperature throughout the week, as well as the indoor temperature after HVAC control; Bottom: the absolute values of discomfort for four users within the optimal indoor temperature range throughout one week.}
		\label{HVAC_control}
	\end{center}
\end{figure}

Next, we analyzed individual discomfort for users based on varying thermal comfort tolerance $\Delta_i$, as shown in Fig. \ref{individual_dis}. In the range $\Delta_i \in [0, 0.5]$, increasing $\Delta_i$ led to a gradual decrease in discomfort for $4$ users, with user $1$ experiencing the highest and user $4$ the lowest discomfort. For $\Delta_i \in [0.5, 1]$, discomfort initially decreased and then increased for all users. In the range $\Delta_i \in [1, 1.5]$, user $2$ consistently had the lowest discomfort, reaching a minimum at $\Delta_i = 1.5$. From $\Delta_i \in [1.5, 3]$, discomfort for all users significantly increased with $\Delta_i$. The blue solid curve in Fig. \ref{individual_dis} represents the maximum discomfort of users at different $\Delta_i$. Notably, discomfort fluctuations were relatively small between $\Delta_i = 0$ and $\Delta_i = 1.5$, indicating stability in user discomfort. Beyond $\Delta_i = 1.5$, discomfort markedly increased. At $\Delta_i = 1.5$, the maximum discomfort among users was only $0.26$, whereas at $\Delta_i = 3$, it approached $0.69$. Thus, a higher thermal comfort tolerance doesn't necessarily lead to better outcomes for users.

Finally, we simulated indoor temperature variations over a week controlled by the HVAC system in response to changing outdoor temperatures. In Fig. \ref{HVAC_control}, we set the weekly minimum temperature range to $[9, 13]$ and the maximum range to $[20, 25]$. Outdoor temperatures for each day were randomly generated. For $4$ users, based on the optimal indoor temperature range obtained from Fig. \ref{min_max_temp}, we divided the settings for each day of the week into different stages. Specifically, the optimal indoor temperature for the first and second days is set to $[18, 19.5]$, for the third to fifth days is set to $[17, 20]$, and for the last two days is set to $[18.4, 18.6]$. These settings correspond to different comfort tolerance levels, namely, $\Delta_i =\Delta =0, 3, 1.5$. Additionally, the heat capacity coefficient $c = 0.1$, and the control gain $k = 1$. In the indoor temperature controlled by the HVAC system, it can be observed that when \(\Delta_i=\Delta = 1.5\), the overall discomfort of the four users reaches the minimum, with user $4$ experiencing a maximum discomfort of less than $0.3$. However, when \(\Delta_i=\Delta = 3\), the disparity in discomfort among users is the greatest.

\section{CONCLUSIONS AND FUTURE WORK}
In standard HVAC comfort settings, despite high energy consumption, users often face challenges related to thermal comfort. Our focus is on developing innovative thermal comfort control technology, addressing users' comfort issues through temperature control while aiming to reduce energy consumption. This paper introduces a novel thermal comfort model and proposes a control approach using direct user feedback signals on comfort/discomfort. The integration of HVAC control in the energy management system with individual comfort preferences can harmonize users' personalized thermal needs while simultaneously reducing energy consumption. Through this method, we determine the optimal indoor temperature range to minimize user discomfort.

Future works could focus on three aspects. Firstly, it is possible to strategically optimize individual thermal comfort models by integrating them with advanced temperature control techniques such as model predictive control. Secondly, integrating thermal comfort control with personal comfort systems holds the potential to further customize and personalize thermal comfort strategies, enhancing overall efficiency and user satisfaction. Furthermore, since each individual can only determine their own thermal comfort tolerance and is unable to ascertain the thermal comfort tolerance of other users, analyzing user decisions in choosing optimal thermal comfort tolerance from a game theory perspective could maximize the overall comfort of the users.

\bibliographystyle{IEEEtran}
\bibliography{reference}

\end{document}